\documentclass{llncs}
\usepackage{fullpage}
\usepackage{graphicx}
\usepackage[utf8]{inputenc}
\usepackage{amsfonts}
\usepackage{amssymb}
\usepackage{amsmath}
\usepackage{xspace}
\usepackage{multirow}
\usepackage{etoolbox}		% kvuli \AtEndEnvironment{}
\usepackage{caption}
\usepackage{pgfplots,tikz}
\usetikzlibrary{angles,quotes}

\usepackage{todonotes}

\newcommand{\ve}[1]{\mathchoice{\mbox{\boldmath$\displaystyle\bf#1$}}
{\mbox{\boldmath$\textstyle\bf#1$}}
{\mbox{\boldmath$\scriptstyle\bf#1$}}
{\mbox{\boldmath$\scriptscriptstyle\bf#1$}}}

\newcommand{\vea}{{\ve a}}
\newcommand{\vealpha}{{\ve \alpha}}
\newcommand{\veb}{{\ve b}}

\newcommand{\veg}{{\ve g}}

\newcommand{\vel}{{\ve l}}

\newcommand{\vep}{{\ve p}}

\newcommand{\veu}{{\ve u}}

\newcommand{\vew}{{\ve w}}
\newcommand{\vex}{{\ve x}}
\newcommand{\vey}{{\ve y}}
\newcommand{\vez}{{\ve z}}

\newcommand{\A}{A^{(n)}}
\newcommand{\R}{\mathbb{R}}
\newcommand{\Z}{\mathbb{Z}}
\newcommand{\N}{\mathbb{N}}

\newcommand{\Oh}{{\cal O}}

\DeclareMathOperator{\poly}{{\rm poly}}
\newcommand{\FPT}{{\sf FPT}\xspace}
\newcommand{\XP}{{\sf XP}\xspace}
\newcommand{\NP}{{\sf NP}\xspace}
\newcommand{\W}[1]{{\sf W}[#1]\xspace}

\AtEndEnvironment{proof}{\qed}

\begin{document}

\title{Scheduling meets $n$-fold Integer Programming\thanks{This research
was partially supported by the project 14-10003S of GA \v{C}R and project 1784214 GA UK.}}



\author{Dušan Knop  \and Martin Koutecký
}


\institute{Department of Applied Mathematics (KAM), \\
	        Charles University in Prague, Czech Republic.\\
	        \email{\{knop, koutecky\}@kam.mff.cuni.cz} }

\maketitle

\begin{abstract}
Scheduling problems are fundamental in combinatorial optimization. Much work has been done on approximation algorithms for \NP-hard cases, but relatively little is known about exact solutions when some part of the input is a fixed parameter. In 2014, Mnich and Wiese initiated a systematic study in this direction.

In this paper we continue this study and show that several additional cases of fundamental scheduling problems are fixed parameter tractable for some natural parameters. Our main tool is $n$-fold integer programming, a recent variable dimension technique which we believe to be highly relevant for the parameterized complexity community. This paper serves to showcase and highlight this technique.

Specifically, we show the following four scheduling problems to be fixed-parameter tractable, where $p_{max}$ is the maximum processing time of a job and $w_{max}$ is the maximum weight of a job:
\begin{itemize}
\item Makespan minimization on uniformly related machines ($Q||C_{max}$) parameterized by $p_{max}$,
\item Makespan minimization on unrelated machines ($R||C_{max}$) parameterized by $p_{max}$ and the number of kinds of machines (defined later),
\item Sum of weighted completion times minimization on unrelated machines ($R||\sum w_jC_j$) parameterized by $p_{max}+w_{max}$ and the number of kinds of machines,
\item The same problem, $R||\sum w_jC_j$, parameterized by the number of distinct job times and the number of machines.
\end{itemize}

%\keywords{Fixed parameterized tractability \and Scheduling on parallel machines}

%\subclass{90B35 \and 90C10 \and 03D15}
%\dkcom{90B35 -- Scheduling theory, deterministic; 90C10 -- Integer programming; 03D15 -- Complexity of computation (including implicit computational complexity)}
\end{abstract}

\section{Introduction}
Scheduling problems are one of the fundamental classes of problems in combinatorial optimization since 1960s \cite{Allahverdi:15,LawlerLKS:93,Potts:09} and many variants of scheduling turn out to be \NP-hard. In response to this one can either look for an approximate solution, or restrict the input in a certain way. Approximation algorithms for scheduling have been an established area of research for a long time now~\cite{LawlerLKS:93}. On the other hand, parameterizing the input in order to obtain exact results has not been studied much before. We say that a problem $P$ with input of size $n$ is \textit{fixed-parameter tractable} (\FPT) with respect to parameter $k$ if there exists a computable function $f$ which does not depend on $n$ and an algorithm solving $P$ with running time $f(k)\poly(n)$; we call an algorithm with this running time \FPT algorithm. Regarding scheduling, Mnich and Wiese~\cite{MnichW:14} have recently initiated a systematic study of the relationship of various scheduling problems and their parameterizations, proving both positive and negative results. In this paper we continue in this direction, examining three additional fundamental scheduling problems and their parameterizations, and devising \FPT algorithms for them.

However, our goal is not merely to prove new positive results. In their work, Mnich and Wiese rely on mathematical programming techniques in \textit{fixed} dimension, which have been introduced in 1983 by Lenstra~\cite{Lenstra:83} and significantly extended in 2000 by Khachiyan and Porkolab~\cite{KhachiyanP:00}. These techniques are by now well established in the \FPT community, even though the power of the latter extension of Khachiyan and Porkolab has not been fully utilized yet, as we will discuss further on. Independently of this, a new theory of \textit{variable} dimension optimization has been developed in the past 15 years; see Onn's book~\cite{Onn:10}. A breakthrough result is an \FPT algorithm for the so-called $n$-fold integer programming ($n$-fold IP) by Hemmecke, Onn and Romanchuk~\cite{HemmeckeOR:13}. In contrast to the fixed dimension techniques, $n$-fold IP is not yet established as an indispensable part of an \FPT researchers toolbox. In this paper we would like to help change that.

Let us now introduce $n$-fold IP. Given $nt$-dimensional integer vectors $\veb, \veu, \vel, \vew$, $n$-fold integer programming ($n$-fold IP) is the following problem in variable dimension $nt$:
\begin{equation}\label{NIP}
\min\left\{\vew\vex\,\colon\A\vex=\veb\,,\
\vel\leq\vex\leq\veu\,,\ \vex\in\Z^{nt}\right\}\ ,
\end{equation}

where
\begin{equation*}\label{NFold}
A^{(n)}\quad:=\quad
\left(
\begin{array}{cccc}
  A_1    & A_1    & \cdots & A_1    \\
  A_2    & 0      & \cdots & 0      \\
  0      & A_2    & \cdots & 0      \\
  \vdots & \vdots & \ddots & \vdots \\
  0      & 0      & \cdots & A_2    \\
\end{array}
\right)\quad
\end{equation*}
is an $(r+ns)\times nt$ matrix with $A_1$ an $r\times t$
matrix and $A_2$ an $s\times t$
matrix. Let $a$ be the biggest absolute value of a number in $A^{(n)}$. The vector $\vex$ is naturally partitioned into $n$ \textit{bricks} of size $t$, that is, we index it as $\vex = (x^1_1, x^1_2, \dots, x^1_t, \dots, x^n_1, \dots, x^n_t)$. As such, $n$-fold IP is best suited for \textit{multi-index} problems whose IP formulation has variables indexed by $[n] \times [l_1] \times \dots \times [l_k]$ for some integers $n, l_1, \dots, l_k$ such that $l_1, \dots, l_k$ are fixed parameters and only $n$ is variable.

Hemmecke, Onn and Romanchuk~\cite{HemmeckeOR:13} prove that there is an \FPT algorithm solving problem (\ref{NIP}) with parameters $r,s,t$ and $a$. We will state and extend their result together with further observations in Section~\ref{sec:nfoldip} 

\subsection{Our contribution}

We consider three non-preemptive scheduling models of increasing generality: parallel identical, uniformly related and unrelated machines (in the standard notation~\cite{LawlerLKS:93} denoted by $P, Q$ and $R$, respectively), and the two most common objective functions: minimizing makespan and sum of weighted completion times (denoted $C_{max}$ and $\sum w_jC_j$, respectively).

For \textit{identical} machines, the problem consists of a set of $n$ \textit{jobs} $J = \{J_1, \dots, J_n\}$ and $m$ \textit{machines} $M = \{M_1, \dots, M_m\}$, and each job $J_j$ has a \textit{processing time} $p_j \in \N$. For \textit{uniformly related} machines, we additionally have for each machine $M_i$ its \textit{speed} $s_i \in \N$, such that processing job $J_j$ on machine $M_i$ takes time $p_j / s_i$. For \textit{unrelated} machines, we have for each job $J_j$ an $m$-dimensional vector $\vep=(p_j^1, \dots, p_j^m)$, $p_j^i \in \N \cup \{\infty\}$ for all $i$, such that processing job $J_j$ on machine $M_i$ takes time $p_j^i$ (in case $p_j^i = \infty$, $J_j$ cannot be executed on $M_i$). We also consider a restricted variant of the unrelated machines model where there are $K$ \textit{kinds} of machines and the vector of processing times for a job $J_j$ is given with respect to kinds of machines: $\vep = (p_j^1, \dots, p_j^K)$, such that processing $J_j$ on machine $M_i$ of kind $k$ takes time $p_j^k$.
Additionally, for the sum of weighted completion times objective, we are given for each job $J_j$ its \textit{weight} $w_j \in \N$.

A \textit{schedule} is an assignment of jobs to machines and times, such that every machine is executing at most one job at any time. For a job $J_j$ we denote by $C_j$ its \textit{completion time}, that is, the time $J_j$ finishes. In makespan minimization, the goal is to minimize $C_{max} = \max_{J_j \in J} C_j$. When minimizing the sum of weighted completion times, the goal is to minimize $\sum_{J_j \in J} w_jC_j$. For example, the problem of minimizing makespan in the identical machines model is denoted $P || C_{max}$. 

The parameters we consider are the following:
\begin{itemize}
\item $p_{max}$: the maximum processing time of any job,
\item $w_{max}$: the maximum weight of any job,
\item $m$: the number of machines,
\item $\theta$: the number of distinct job processing times and weights (in case of the $\sum w_j C_j$ objective); note that $\theta$ generalizes parameter $p_{max}$,
\item $K$: the number of kinds of machines (defined above).
\end{itemize}

In all of the cases we consider, we use such a combination of parameters that the number $\Theta$ of distinct job \textit{types} is bounded, where jobs of a given type are indistinguishable from each other. This means that the set of jobs $J$ on input can be given compactly by specifying integers $n_1, \dots, n_{\Theta}$ with $n = n_1 + \dots + n_\Theta$, such that $n_j$ denotes the number of jobs of type $j$ that are to be scheduled. Modeling our terminology after Onn~\cite{OnnS:15}, we call a problem \textit{huge} when the numbers $n_j$ on input are given in binary. For example, the \textsc{Cutting Stock} problem can be seen as the huge variant of the \textsc{Bin Packing} problem. All of our results work for the huge variant. 

We show that:

% \begin{theorem}\label{thm:qcmax}
% It is possible to solve $Q || C_{max}$ problem in time $\Oh(p^{p^2})\poly(n)$, where $p=p_{max}$ is the largest processing time of any job and thus, it admits an \FPT algorithm parameterized by $p_{max}$.
% \end{theorem}

\begin{theorem}\label{thm:nfold_sched}
The following scheduling problems are \FPT with respect to parameter $\Theta$ as defined below and solvable in time $\Theta^{\Oh(\Theta^2)}n^{\Oh(1)}$, where 
\begin{enumerate}
\item $Q || C_{max}$: $\Theta = p_{max}$ \label{thm:qcmax_it}
\item $R || C_{max}$: $\Theta = p_{max}^K$ \label{thm:rcmax_it}
\item $R || \sum w_j C_j$: $\Theta = (\max \{p_{max}, w_{max}\})^K$ \label{thm:rwici_it}
\end{enumerate}
\end{theorem}

\begin{theorem}\label{thm:rwicim}
$R || \sum w_j C_j$ is \FPT parameterized by $\Theta = m\theta^m$ and solvable in time $\Theta^{\Oh(\Theta)} n^{\Oh(1)}$.
\end{theorem}

Note that part \eqref{thm:qcmax_it} of Theorem~\ref{thm:nfold_sched} for the easier $P || C_{max}$ problem was already shown by Mnich and Wiese~\cite{MnichW:14}. However, our approach is substantially different and more straightforward, as demonstrated by the immediate extension to $Q || C_{max}$. Also, our result only has a single-exponential dependence on the parameter, unlike the double-exponential dependence of Mnich and Wiese. Thus, Theorem~\ref{thm:nfold_sched} serves to highlight the usefulness of $n$-fold integer programming in parameterized complexity. Theorem~\ref{thm:rwicim} differs as it is proved using the well known fixed dimension techniques of Khachiyan and Porkolab~\cite{KhachiyanP:00}.

In order to prove part \eqref{thm:rwici_it} of Theorem~\ref{thm:nfold_sched} we use an $n$-fold IP formulation and optimize a separable convex function over it. However, the algorithm of Hemmecke et al.~\cite{HemmeckeOR:13} only works for linear and certain restricted separable convex objectives. Thus, using the ideas of Hemmecke, Köppe and Weismantel~\cite{HemmeckeKW14}, we extend the previous result to optimizing any separable convex function.

We complement our positive findings by two hardness results:

\begin{theorem}\label{thm:sched_hardness}
$P || C_{max}$ and $P || \sum w_j C_j$ is \W{1}-hard when parameterized by $m$, even when job processing times and weights are given in unary.
\end{theorem}

\subsection{Related work}
We give a brief summary of known results related to scheduling and parameterized complexity. Many more results can be found in surveys on this topic (e.g.~\cite{LawlerLKS:93}). Note that these surveys focus on \NP-hardness results and polynomial and approximation algorithms. There were several attempts to introduce scheduling problems to the \FPT community. It started with the pioneering work of Bodlaender and Fellows~\cite{BodlaenderF95} for the precedence constrained scheduling and continued (after nearly 10 years) with the first \FPT algorithm of Fellows and McCartin~\cite{FellowsM03}. A recent result of van Bevern et al.~\cite{BevernBBKTW16} resolves a question of Mnich and Wiese by showing that makespan minimization with precedence constraints $P | prec | C_{max}$ is \W{2}-hard. Marx~\cite{Marx:09Dagstuhl,Marx:11Dagstuhl} also highlighted the importance of scheduling in the parameterized setting.

Many other settings are now popular in the scheduling community. These include {\em Open Shop}, where we are given for each job $J$ a bunch of tasks $T^J_1, T^J_2,\ldots, T^J_k$ with associated machines $M^J_1, M^J_2,\ldots, M^J_k$ to be completed in the natural order (that is $T^J_i$ has to finish before $T^J_{i+1}$ starts to be processed) -- see related work of van Bevern and Pyatkin~\cite{BevernP16} and Kononov et al.~\cite{KononovSS12}. In {\em Two agent scheduling} there are two agents competing for one machine to schedule their respective jobs within each agents budget. Hermelin et al.~\cite{HermelinKSTW15} showed that if one agent has only $k$ jobs and the other agent has jobs with unit processing times then the problem admits an \FPT algorithm with parameter $k$. Halldórsson and Karlsson~\cite{HalldorssonK06}, followed by work of van Bevern et al.~\cite{DBLP:journals/scheduling/BevernMNW15,BevernNS15} investigate {\em interval scheduling} (or \textit{job interval selection}) in which each job has several time slots in which it may be processed and the task is to schedule as many jobs as possible.

We now turn our attention towards more classical models of scheduling. A summary of what follows can be found in Table~\ref{tab:ComplexityOverview}.

\paragraph{Makespan and identical machines -- $P || C_{max}$}
Most importantly, Mnich and Wiese~\cite{MnichW:14} show that there is an \FPT algorithm for this problem when parameterized by $p_{max}$. A remarkable result of Jansen et al.~\cite{JansenKMS:13} shows that the \textsc{Unary Bin Packing} problem is \W{1}-hard when parameterized by the number of bins. This immediately implies the \W{1}-hardness of $P || C_{max}$ parameterized by $m$ even when $p_{max}$ is given in unary, and with some effort also implies the hardness of $P || \sum w_j C_j$ with parameter $m$ when $p_{max}$ and $w_{max}$ are given in unary, as we show in Section~\ref{sec:hardness}.

\paragraph{Makespan and unrelated machines -- $R || C_{max}$}
%\begin{itemize}
%\item input $n$ (unary), parameter ($\theta$, $m$) $\Rightarrow$ \FPT. \cite{MnichW:14}
%\item input $n,m$, $\theta$ constant $\Rightarrow$ \NP-h
%\item input $n$, $\theta$, $m$ constant $\Rightarrow$ \NP-h
%\item $R || \sum C_j$ is solvable in polynomial time (Bruno et al.~\cite{Bruno:74}, Horn~\cite{Horn:73})
%\item Asahiro et al.~\cite{AsahiroJMOZ:07} show that $R || C_{max}$ is strongly \NP-h already for \textit{restricted assignment} (i.e., there is a number $p_j$ for each job such that for each machine $p_{ij} \in \{p_j, \infty\}$) and all $p_j \in {1,2}$ and for every job there are exactly two machines where it can run.
%\end{itemize}
Asahiro et al.~\cite{AsahiroJMOZ:07} show that the problem $R || C_{max}$ is strongly \NP-hard already for \textit{restricted assignment} when there is a number $p_j$ for each job such that for each machine $p_j^i \in \{p_j, \infty\}$ and all $p_j \in \{1,2\}$ and for every job there are exactly two machines where it can run. Mnich and Wiese~\cite{MnichW:14} proved that the problem is in \FPT with parameters $\theta$ and $m$.
 
\paragraph{Sum of weighted completion times and unrelated machines -- $R || \sum w_j C_j$}

Surprisingly, in the unweighted case, $R || \sum C_j$ turns out to be solvable in polynomial time~\cite{Bruno:74,Horn:73}. Preemption ($R | pmtn | \sum C_j$) makes the problem strongly \NP-hard~\cite{Sitters:05}. The weighted case $R || \sum w_j C_j$ is strongly \NP-hard~\cite[Problem SS13]{GareyJ:79}.

\begin{table*}
\def\arraystretch{1.1}
\centering
\begin{tabular}{| c | cccc | c |}
\hline
Problem & $n$ & $m$ & $p_{max}$ & $\theta$ & Complexity \\
\hline\hline
\multirow{3}{*}{$P || C_{max}$} &
unary & unary & unary & --- & \NP-hard \cite{JansenKMS:13} \\
%\hline
& unary & binary & param. & --- & \FPT~\cite{MnichW:14}, $2^{2^{\Oh(p_{max}^2 \log p_{max})}}$ \\
%\hline
& unary & param. & unary & --- & \W{1}-hard~\cite{JansenKMS:13} \\
\hline
$Q || C_{max}$ & binary & unary & param. & --- & \FPT [*], $2^{\Oh(p_{max}^2 \log p_{max})}$ \\
\hline
\multirow{4}{*}{$R || C_{max}$} & binary & param. & --- & param. & \FPT~\cite{MnichW:14} \\
& unary & unary & --- & constant & \NP-hard~\cite{AsahiroJMOZ:07} \\
& binary & constant & --- & binary & \NP-hard \\
& binary & unary & --- & param. & \FPT with $K$ [*] \\
\hline
\multirow{2}{*}{$R || \sum w_jC_j$}
& binary & unary & param. & --- & \FPT with $w_{max}, K$ [*] \\
& binary & param. & --- & param. & \FPT [*] \\
\hline
\end{tabular}
 \caption{A summary of the complexity results we mention. The results contained in this paper are marked with [*].}\label{tab:ComplexityOverview}
\end{table*}

\section{Preliminaries}

For a positive integer $n$, we denote $[n] = \{1, 2, \dots n\}$ and $\langle n \rangle = \lceil \log_2(n) \rceil$ the length of binary encoding of $n$. We write vectors in bold such as $\veb = (b_1, \dots, b_n)$. By $\langle \veb \rangle$ we denote the length of encoding of $\veb$ which is $\sum_{i=1}^n \langle b_i \rangle$; similarly for matrices. Let $\vex = (x_1, \dots, x_n)$. A function $f: \R^n \rightarrow \R$ is \textit{separable convex} if it can be written as $f(\vex) = \sum_{i=1}^n f_i(x_i)$ such that $f_i$ is univariate convex for all $1 \leq i \leq n$. Given a function $f$, a \textit{comparison oracle}, queried on two vectors $\vex,\vey$, asserts whether or not  $f(\vex)\leq f(\vey)$. Time complexity measures the number of arithmetic operations and oracle queries.

\subsection{$N$-fold Integer Programming}
\label{sec:nfoldip}

Recall the definition of the $n$-fold IP problem~(\ref{NIP}). Observe that in an $n$-fold IP, every equality $\vea \vex = b$ of $A^{(n)}\vex = \veb$ is of one of two kinds. Either $\vea = (\vealpha, \vealpha, \dots, \vealpha)$ corresponds to a certain row $\vealpha$ of $A_1$ repeated $n$ times, meaning that $a_i^j = \alpha_i$ for all $1 \leq j \leq n$; we call this equality \textit{globally uniform}. Or, $\vea = (\ve0, \dots, \ve0, \vealpha, \ve0, \dots, \ve0)$, corresponding to $kt$ zeros, a certain row $\vealpha$ of $A_2$, and $(n-k-1)t$ zeros, such that there also must exist $n-1$ other equalities of this form which have the same coefficients $\vealpha$ on the remaining $n-1$ bricks and zeros elsewhere. We call this kind of constraint \textit{locally uniform}. Thus, given an IP in dimension $nt$, we can prove that it is an $n$-fold IP by showing that every equality is either globally uniform or locally uniform.

Let $L:=\langle a,\veb,\vel,\veu,\vew \rangle$ be the length of the input. The first key result we use is the following:

\begin{theorem}{(\cite[Theorem 6.1]{HemmeckeOR:13})}\label{thm:nfoldip}
For any fixed $r,s$ and $t$, there is an algorithm that,
given $n$, $a$, $(r,s)\times t$, matrices $A_1$ and $A_2$ of appropriate dimensions with all entries
bounded by $a$ in absolute value, $\veb,\vel,\veu$, and $\vew$, solves problem \eqref{NIP} in time $\Oh(a^{3t(rs+st+r+s)}n^3L)$.

%$$\min\left\{ \vew \vex\,\colon\A\vex=\veb\,,\
%\vel\leq\vex\leq\veu\,,\ \vex\in\Z^{nt}\right\}.$$
\end{theorem}

In other words, if $r,s,t$ and $a$ are parameters and $\A, \veb, \vel, \veu, \vew$ are input, the $n$-Fold IP problem can be solved in \FPT time.

\subsection{Separable convex minimization}

We also need a result regarding minimization of separable convex function. Given a separable convex function $f$, we consider integer programs of the form:

\begin{equation}\label{SIP}
\min\left\{f(\vex)\,\colon \A\vex=\veb\,,\
\vel\leq\vex\leq\veu\,,\ \vex\in\Z^{nt}\right\}.
\end{equation}

Hemmecke, Onn and Romanchuk prove that:

\begin{theorem}{(\cite[Theorem 4.1]{HemmeckeOR:13})}\label{thm:nfoldip_convex_step}
For matrices $A_1$ and $A_2$ of appropriate dimensions, there is an algorithm that, given $n$,
$\veb,\vel,\veu$, separable convex function $f$ presented by a comparison oracle,
and a feasible point $\vex$ in the program (\ref{SIP}), either asserts
that $\vex$ is optimal or finds an augmenting step $\veg$
for $\vex$ which satisfies $f(\vex+\veg) < f(\vex)$ in linear time $\Oh(n)$.
\end{theorem}

In particular, if we are able to find an initial solution $\hat{\vex}$ and guarantee that the optimum $\vex^*$ is near, that is, $||\hat{\vex} - \vex^*||_\infty \leq N$ for some $N \in \N$, applying Theorem~\ref{thm:nfoldip_convex_step} at most $nN$ times will reach the optimum.

Under suitable assumptions on the function $f$, the continuous optimum $\hat{\vex}$ of Problem~\ref{SIP} can be found in polynomial time using the ellipsoid method or an interior point method. The key insight of Hemmecke, Köppe and Weismantel~\cite{HemmeckeKW14} is adapting a proximity technique of Hochbaum and Shantikumar~\cite{HochbaumS:90} for the context of Graver bases:

\begin{theorem}{(\cite[Theorem 3.14]{HemmeckeKW14})}
Let $\hat{\vex}$ be an optimal solution of the continuous relaxation
of \eqref{SIP},

$$ \min\left\{f(\vex)\,\colon \A\vex=\veb\,,\
\vel\leq\vex\leq\veu\,,\ \vex\in\R^{nt}\right\} \enspace . $$

Then there exists an optimal solution $\vex^*$ of the integer optimization problem \eqref{SIP} with

$$ ||\hat{\vex} - \vex^*||_\infty \leq n \cdot \max \{||v||_\infty \mid v \in \mathcal{G}(A)\} \enspace .$$ 
\end{theorem}

Here, $\mathcal{G}(A)$ is the Graver basis of the bimatrix $A$, and the quantity $\max \{||v||_\infty \mid v \in \mathcal{G}(A)\}$ is bounded by a $t \cdot a^{r+s}$ (\cite[Lemma 3.20]{Onn:10}). Hence, $n \cdot ||\hat{\vex} - \vex^*||_\infty \leq n^2 \cdot t \cdot a^{r+s}$ applications of Theorem~\ref{thm:nfoldip_convex_step} (with $\vex = \lfloor{\hat{\vex}}\rfloor$) reach the integer optimum:

\begin{theorem}\label{thm:nfoldip_convex}
For any fixed $r,s$ and $t$, there is an algorithm that,
given $n$, $a$, $(r,s)\times t$, matrices $A_1$ and $A_2$ of appropriate dimensions with all entries
bounded by $a$ in absolute value, $\veb,\vel,\veu$, and $\vew$, solves problem~\eqref{SIP} in
time $\Oh(a^{3t(rs+st+r+s)}n^3L)$.
\end{theorem}

\subsection{Convex minimization in fixed dimension}

To prove Theorem~\ref{thm:rwicim} we formulate an integer linear program in fixed dimension and minimize a convex function over it:

\begin{theorem}{(Khachiyan and Porkolab~\cite{KhachiyanP:00})}\label{thm:semialgebraic_ip_fpt}
  Minimizing a quasiconvex function over a semialgebraic convex set
  defined by $k$ polynomials in dimension $p$ is \FPT with respect to
  $k$ and $p$.
\end{theorem}

In order not to delve into unnecessary details, let us say that a quasiconvex function is a generalization of a convex function, and semialgebraic convex sets are a generalization of convex sets, containing e.g. the feasible regions of semidefinite programs (SDPs); see the book by Blekherman et al.~\cite{BlekhermanPT:12}. There are less general variants of this result which attain better running times, e.g. by Hildebrand and Köppe~\cite{HildebrandK13}. Since our aim is to simply prove fixed-parameter tractability, we choose to state the most general result.

Note that we are not aware of an application of Theorem~\ref{thm:semialgebraic_ip_fpt} which would use the fact that one can optimize over a \textit{region} more general than the integer hull of a polyhedron (i.e., a region given by non-linear convex constraints). There is a ``linearization trick'' which is widely used (including by Mnich and Wiese): a convex constraint $\vea \vex \leq f(x_i)$ whose domain is bounded by some number $N$ given in unary can be rewritten as $N$ linear constraints describing the piecewise linear approximation of $\vea \vex \leq f(x_i)$ which is exact on the $N$ integers of its domain. Then, the feasible region is the integer hull of a polyhedron. To the best of our knowledge, there is no result whose feasible region is given by a set of constraints that cannot be ``linearized'' as described above. So both us and Mnich and Wiese~\cite{MnichW:14} only need this result because of its generality in terms of the \textit{objective function}, not the \textit{feasible region}. It is an interesting open problem to find an application of Theorem~\ref{thm:semialgebraic_ip_fpt} whose feasible region is given by, for example, a semidefinite program in fixed dimension.

\subsection{Smith's rule -- structure of solutions when minimizing $\sum w_jC_j$}

Here we make a few basic observations about the structure of (optimal) solutions in the problem $R || \sum w_jC_j$. To do so, we utilize a useful way of visualizing the objective function $\sum w_jC_j$ called \textit{two-dimensional Gantt charts} (2D Gantt charts), which was introduced by Goemans and Williamson~\cite{GoemansW:00}.

Let us first introduce the following notation. Fix a machine $M_i$ and assume that the set of jobs scheduled to run on $M_i$ is given, denoted $J^i$. Whenever the index of a machine $i$ is clear from context, we omit it. For any set of jobs $S \subseteq J$ let $w(S) = \sum_{J_j \in S} w_j$ and $p(S) = \sum_{J_j \in S} p_j$. A 2D Gantt chart starts at point $(0, w(J^i))$ and ends at $(p(J^i), 0)$. Each job $J_j \in J^i$ is represented by a rectangle of length $p_j$ and height $w_j$ whose position is defined by a startpoint and an endpoint. The startpoint $(t,w)$ of a job is the endpoint of a previous job (or $(0, w(J^i))$ for the first job) while its endpoint is $(t+p_j, w-w_j)$. The value $\sum_{J_j \in J^i} w_j C_j$ is then simply the area beneath the upper sides of the rectangles; see Figure~\ref{fig:2dgantt}.

\begin{figure}[!bt]
  \includegraphics{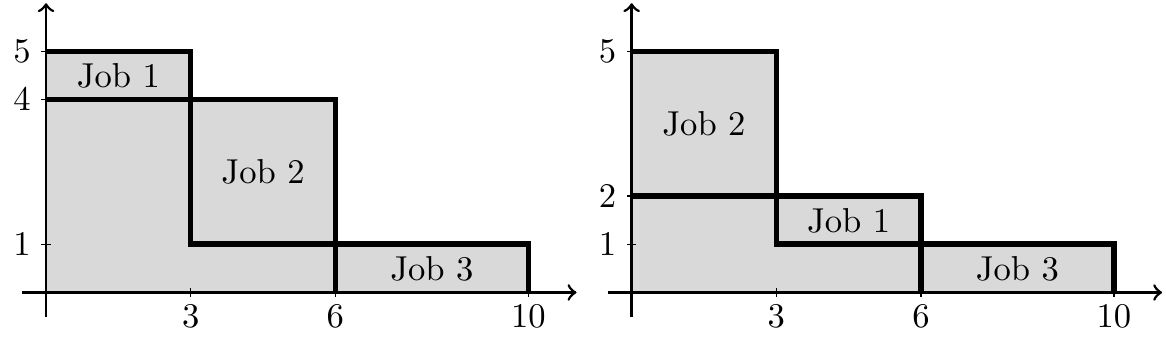}
  \caption{An example of a 2D Gantt chart with three jobs of lengths 3, 3 and 4 and weights 1, 3 and 1, respectively. To the left is one possible ordering, to the right is the ordering given by Smith's rule, producing an optimal schedule (minimizing the gray area).}\label{fig:2dgantt}
\end{figure}

For $\sum w_j C_j$ minimization on one machine with no precedence constraints (no restrictions on the order of the jobs) there is a simple observation about the structure of an optimal schedule:

\begin{lemma}{(Smith's rule~\cite{GoemansW:00})}
Given a set of jobs $J^i$, a schedule minimizing $\sum_{J_j \in J^i} w_j C_j$ is given by ordering the jobs by non-increasing $\rho_i(j) = w_j/p_j^i$.
\end{lemma}

Since the ratios $\rho_i(j)$ correspond to slopes of the rectangles in a 2D Gantt chart, Smith's rule implies that the chart of an optimal schedule will have slopes which form a piecewise linear convex function. Goemans and Williams then go on to observe that for such a chart there is an alternate way of computing its area based on splitting it into triangles; see Figure~\ref{fig:smith}. That leads us to this lemma:

\begin{figure}[bt]
\captionsetup{singlelinecheck=off}
%  \begin{minipage}{.4\textwidth}
   % \input{tikz/triangles.tikz}
    \includegraphics{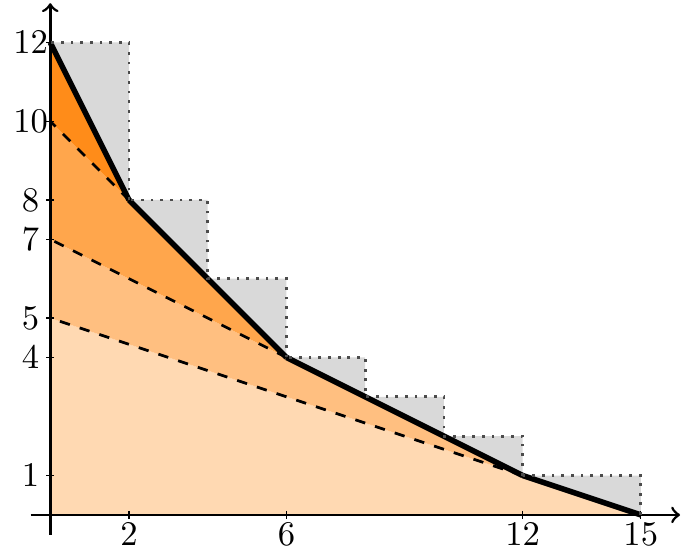}
%  \end{minipage}
%  \begin{minipage}[b]{.6\textwidth}
\caption[Bla bla]{An alternate way of computing the objective function $\sum_{J_j \in J^i} w_j C_j$. In this examples there are following jobs:
    \begin{itemize}
    \item 1 job with processing time $2$ and weight $4,$
    \item 2 jobs with processing time $2$ and weight $2,$
    \item 3 jobs with processing time $2$ and weight $1,$
    \item 1 job with processing time $3$ and weight $1.$
    \end{itemize}
    }\label{fig:smith}
 % \end{minipage}
\end{figure}

\begin{lemma}\label{lem:chartArea}
Given jobs $J^i = \{J_1, \dots, J_l\}$ scheduled to run on machine $M_i$ such that $\rho(j) \geq \rho(j+1)$ for all $1 \leq i \leq l-1$, the optimal schedule has value 
$$\sum_{j=1}^l (\frac{1}{2} p(\{J_1,\dots, J_j\})^2 (\rho(j)-\rho(j+1)) +\frac{1}{2} w_j p_j).$$
\end{lemma}
\begin{proof}
Given the set $J^i$, the optimal schedule on the machine $M_i$ is determined according to Smith's rule. It is possible to divide the area as can be seen on Figure~\ref{fig:smith}. Note that the gray area is determined by the set of jobs $J^i$ and in fact can be computed just from the knowledge of $J^i$ as $\frac{1}{2}w_j p_j$ for each job $J_j\in J^i.$ This results in the linear term in the statement.

It remains to compute the area under the bold line (orange area in color printing) of Figure~\ref{fig:smith} -- i.e. the area under the piecewise linear function (again determined by the observed structure of the solution). We divide the area into triangles and compute the total area as a sum of those.

For a job $J_j$ we compute the contribution of the job as the associated area. The total area can be expressed as a difference of area of two impedance triangles -- see Figure~\ref{fig:geometry} for illustration. The length of the common ordinate parallel to processing time axis is $p(\{J_1,\dots, J_j\}),$ so it remains to establish the height $b$ of the two triangles. We express it with the help of a tangent rule as $b = a\tan\varphi$ (and $b' = a\tan\varphi$)

It is straightforward to express the total area contribution of the job $J_j$ as 
\begin{equation}
\begin{split}
J_j & = \frac{1}{2}p\big(\{J_1,\dots, J_j\}\big)\cdot (b'-b) \\
    & = \frac{1}{2}p\big(\{J_1,\dots, J_j\}\big)\cdot p\big(\{J_1,\dots, J_j\}\big)\big(\rho(j) - \rho(j+1)\big) \\
    & = \frac{1}{2}p\big(\{J_1,\dots, J_j\}\big)^2(\rho(j) -\rho(j+1)).
\end{split}
\end{equation}
%$$
Summing over all jobs in $J^i$ finishes the proof.

\begin{figure}[bt]
%  \begin{minipage}{.5\textwidth}
    %\input{tikz/objective_machine.tikz}
    \includegraphics{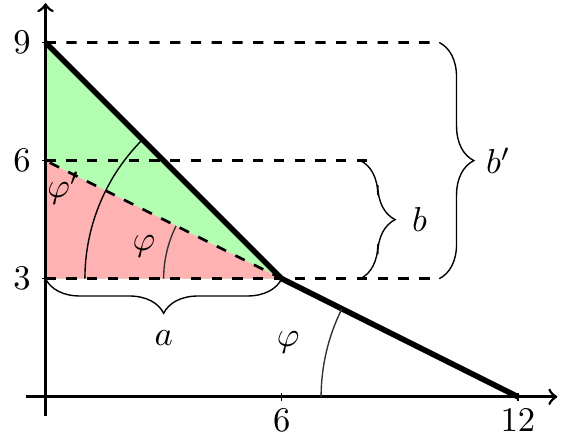}
%  \end{minipage}
%  \begin{minipage}[b]{.5\textwidth}
    \caption{It is possible to compute the area of the lighter gray (green in color printing) rectangle as a difference between the area of an auxiliary triangle ($\frac{1}{2}a\cdot b'$ in the figure) and the dark gray (red in color printing) triangle. Note that it is possible to compute $b,b'$ from the value of $a$ and the tangent of $\varphi, \varphi'.$}\label{fig:geometry}
%  \end{minipage}
\end{figure}
\end{proof}

In our setup the set $J^i$ will be given by integers $x_1^i, \dots, x_\Theta^i$ representing how many jobs of each type are scheduled to run on machine $M_i$. Observe that for each type $1 \leq j\leq\Theta,$ jobs of type $j$ have identical slope $\rho(j)$ and thus correspond to a single triangle in the chart. We have the following corollary:

\begin{corollary}\label{cor:separable_convex}
Given integers $x_1^i, \dots, x_\Theta^i$ representing numbers of jobs of each type scheduled to run on machine $M_i$ and a permutation $\pi_i: [\Theta] \rightarrow [\Theta]$ such that $\rho_i(\pi_i(j)) \geq \rho_i(\pi_i(j+1))$ for all $1 \leq j\leq \Theta-1$, the optimal schedule has value $\sum_{i=1}^\Theta (\frac{1}{2} (z_j^i)^2 + \frac{1}{2} x_j^i p_j^i w_j),$ where $z_j^i = \sum_{l=1}^j p_l^i x_l^i$.
\end{corollary}

\section{\W{1}-hardness of unary $P || C_{max}$ and $P || \sum w_j C_j$}
\label{sec:hardness}

The \textsc{Bin Packing} problem asks whether $k$ bins of size $B$ suffice to contain a set of $n$ items represented by integers $o_1, \ldots, o_n$.
It is straightforward from the \textsc{2-Partition} problem that this is \NP-complete when $o_1, \ldots, o_n$ are large numbers given in binary.
Remarkably, Jansen et al.~\cite{JansenKMS:13} prove that even when the numbers are given in unary (i.e., assume $\max_i o_i \leq n$), there is likely no $f(k)n^{\Oh(1)}$ algorithm as the problem is \W{1}-hard parameterized by $k$. This hardness obviously translates to makespan minimization: deciding whether there is a schedule of jobs $o_1, \ldots, o_n$ on $k$ machines with $C_{max} = B$ is equivalent to deciding the aforementioned \textsc{Unary Bin Packing} problem.

To the best of our knowledge, the analogous question regarding the complexity of $P || \sum w_j C_j$ parameterized by the number of machines $m$ was not yet considered. We prove that it is \W{1}-hard by once again reducing \textsc{Unary Bin Packing} to it.

Jansen et al. note that their hardness result stands even for \textit{tight} instances of \textsc{Unary Bin Packing}, that is, instances where $\sum_i o_i = kB$.
Given a tight instance of \textsc{Unary Bin Packing}, construct an instance of ${P || \sum w_j C_j}$ consisting of $k$ machines and $n$ jobs with $p_j = w_j = o_j$ for $j=1, \ldots, n$. Let $J^\ell$ and $\hat{C}_\ell$ for $\ell=1,\ldots,k$ denote the set of jobs scheduled on machine $\ell$ and the completion time of machine $\ell$, respectively. Because the ratio $\frac{w_j}{p_j}$ is identical for all jobs, the ordering of jobs on a each machine is irrelevant. Thus, the contribution of a machine $\ell$ to the objective function is

$$ \frac{1}{2}\hat{C}_\ell^2 + \sum_{j \in J^\ell} \frac{p_j^2}{2} \enspace .$$

Summing over all machines, we get

$$ \Big( \sum_{j \in J^\ell} \frac{p_j^2}{2} \Big) + \Big( \sum_{\ell=1}^k \frac{1}{2}\hat{C}_\ell^2 \Big)  \enspace .$$

We argue that a schedule with cost $\frac{kB^2}{2} + kB$ exists if and only if the original \textsc{Unary Bin Packing} instance is a ``yes'' instance. Observe that there is only one schedule of this cost, namely one where $\hat{C}_\ell = B$ for all $\ell = 1, \ldots, k$. Let $\Delta_\ell = \hat{C}_\ell - B$ for all $\ell = 1, \ldots, k$. Disregarding the term $\sum_{j \in J^\ell} \frac{p_j^2}{2}$ which is independent of the schedule, the contribution becomes

$$ \sum_{\ell=1}^k \frac{1}{2}(B+\Delta_\ell)^2 = \sum_{\ell = 1}^k \frac{1}{2}(B^2 + 2\Delta_\ell B + \Delta_\ell^2) \enspace .$$

Since $\sum_\ell \Delta_\ell = 0$,  $\sum_\ell 2\Delta_\ell B = 0$. Thus $\sum_\ell \Delta_\ell^2 = 0$ exactly when $\Delta_\ell = 0$ for all $\ell$, that is, when $\hat{C}_\ell = B$ for all $\ell$. That concludes the proof of Theorem~\ref{thm:sched_hardness}.

\section{\FPT results}
\label{sec:scheduling}
\subsection{Warmup: $P||C_{max}$ and $Q||C_{max}$ parameterized by $p_{max}$}

In this section we show that $P||C_{max}$ and its generalization $Q || C_{max}$ is \FPT parameterized by $p_{max}$.

Denote $\Theta:= p_{max}$. We say that a job of length $j$ is of \textit{type} $j$. On input we have $n$ jobs of at most $\Theta$ types, given as numbers $n_1, \dots n_\Theta$ encoding the number of jobs of given type (in binary). For every machine $i$ we have variables $x^i_1, \dots, x^i_{\Theta}$; in the solution the interpretation of $x^i_j$ is ``how many jobs of type $j$ are scheduled on machine $i$''. Let us fix a time $T \in \N$; the IP we will formulate will be feasible if there is a schedule with $C_{max} \leq T$.

To assure that each job is scheduled on some machine, we state these globally uniform constraints:
\begin{equation}\label{eq:pcmax_global}
\sum_{i=1}^m x_j^i = n_j \quad\forall 1 \leq j \leq \Theta.
\end{equation}

To assure that, for every machine $M_i$, $1 \leq i \leq m$, the lengths of jobs scheduled on $M_i$ sum up to at most $T$, we state a locally uniform constraint:
\begin{equation}\label{eq:pcmax_local}
\sum_{j=1}^\Theta j x_j^i \leq T
\end{equation}

(Note here that using inequalities instead of equations does not cause problems: we can simply introduce a slack variable for every inequality: $\sum_{j=1}^\Theta j x_j^i + x_s^i = T,
x_s^i \geq 0$; it is also possible to add a suitable number of unit-processing time jobs and work with equalities directly; similarly in the following.)

Clearly if this program is feasible, then there exists a schedule with $C_{max} \leq T$. Finding minimum $T$ can be then done in polynomially many steps by binary search. Thus we want to show that checking feasibility is \FPT by applying Theorem~\ref{thm:nfoldip}.

To apply Theorem~\ref{thm:nfoldip}, we need to bound the values $r,s,t$ and $a$. Clearly the brick size is $t=\Theta$, the number of globally uniform constraints is $r=\Theta$ and the number of locally uniform constraints per brick is $s=1$. Finally, the largest coefficient is $a=\Theta$.

\paragraph{$\mathbf{Q || C_{max}}$} Now we are given speeds $s_i$ for every machine, such that executing a job with processing time $j$ takes time $j/s_i$. The globally uniform constraints~(\ref{eq:pcmax_global}) are the same, but the locally uniform constraints~(\ref{eq:pcmax_local}) now become:
\begin{equation}
\sum_{j=1}^\Theta j x_j^i \leq s_i T.
\end{equation}

Observe that only the right hand side differs for every machine. This finishes the proof of part \eqref{thm:qcmax_it} of Theorem~\ref{thm:nfold_sched}.

\subsection{$R || C_{max}$ parameterized by $p_{max}$ and $K$}

Now we turn to the unrelated machines model. Observe that with parameters $p_{max}$ and the number of kinds of machines $K$, there are at most $\Theta = (p_{max}+1)^K$ possible vectors $\vep_i$ of processing times with respect to kinds of machines, which we call \textit{types}, and each job is of a certain type. The input is then again given (in binary) by $\Theta$ integers $n_1, \dots, n_\Theta$ specifying the number of jobs of each type.

As before, we will describe an $n$-fold IP solving the problem. We have $n\Theta$ variables $x_j^i$ with the same interpretation as above. The globally uniform constraints are the same as before, (\ref{eq:pcmax_global}).

In the previous examples the locally uniform constraints were used to specify that the jobs assigned to each machine finish by time $T.$ However, now we need to specify a different constraint for each \textit{kind} of machine, which might seem hard to do ``uniformly''. Fortunately, because the number of kinds of machines is bounded, we can actually specify all constraints \textit{simultaneously} and make all but ``the right one'' irrelevant by differing right hand side.

Formally, let $B$ be some number bigger than $n p_{max}$, then for machine $M_i$ which is of kind $k$ we have locally uniform constraints
%$\sum_{i=1}^{\Theta} p_j^{k'} x_i^{k'} \leq B$, for all $1 \leq k' \neq k \leq K$ and $\sum_{i=1}^{\Theta} p_j^{k} x_i^{k} \leq T$.
\begin{alignat*}{2}
\sum_{j=1}^{\Theta} p_j^{k'} x_j^{k'} &\leq B \qquad &\forall 1 \leq k' \neq k \leq K, \\
\sum_{j=1}^{\Theta} p_j^{k} x_j^{k} &\leq T. &
\end{alignat*}
In the above constraints whenever $p_j^k = \infty$ we replace it by zero and forbid the job to be run on this kind of machine by specifying appropriate upper bounds:
\begin{alignat}{3}
x_j^i &\leq n_j \qquad &&\forall 1 \leq k \leq K, p_j^k \in [p_{max}]\label{eq:rcmax_ub}\\
x_j^i &\leq 0   &&\forall 1 \leq k \leq K, p_j^k = \infty,\\
x_j^i &\geq 0   &&\forall i,j\label{eq:rcmax_lb}.
\end{alignat}
Observe that the constraints above are indeed locally uniform -- they are identical for every machine up to the right hand side. Thus we have defined an $n$-fold IP, which is feasible if there is a schedule with $C_{max} \leq T$. It remains to observe that all of $r,s,t$ and $a$ are bounded by our choice of parameters $p_{max}$ and $K$: clearly $t=\Theta$, $r=\Theta$, $s=K$ and $a=p_{max}$. Using Theorem~\ref{thm:nfoldip} concludes the proof of part \eqref{thm:rcmax_it} of Theorem~\ref{thm:nfold_sched}.

\subsection{$R || \sum w_j C_j$ parameterized by $p_{max} + w_{max}$ and $K$}
Here we turn our attention to the $\sum w_j C_j$ objective.

Recall what follows from Corollary~\ref{cor:separable_convex}: we have shown that, given integers $x_1^i, \dots, x_\Theta^i$ representing numbers of jobs of each type to be scheduled on machine $M_i$, the contribution of $M_i$ to the objective function is  $f^i(\vex^i, \vez^i) = \sum_{j=1}^\Theta (\frac{1}{2} (z_j^i)^2 (\rho_i(j)-\rho_i(j+1)) + \frac{1}{2} x_j^i p_j^i w_j)$ where $z_j^i = \sum_{l=1}^j p_l^i x_l^i$ and $\pi_i: [\Theta] \rightarrow [\Theta]$ is a permutation such that $\rho_i(\pi_i(j)) \geq \rho_i(\pi_i(j+1))$ for all $j \in [\Theta -1]$.
Observe that $f^i$ is separable convex and thus also $f = \sum_i f^i.$ 
Our goal now is to once again formulate an $n$-fold IP, however this time we need to introduce new variables $z_{j,k}^i$, for $j \in [\Theta], k \in [K]$ and $i \in [m].$ For a machine $M_i$ of kind $k$, we want $z_j^i = z_{j,k}^i,$ so that we can use the formulation of $f$ which we just stated. Notice that we are introducing many ``unnecessary'' variables $z_{j,k}^i$ for kinds $k' \neq k$. This is in order to have a uniform set of local constraints.

The globally uniform constraints (\ref{eq:pcmax_global}) stay the same. The locally uniform constraints serve to project the brick $\vex^i$ to the variables $z_{j,k}^i,$ with permutations $\pi_k$ as defined above:
$$
\sum_{l=1}^j x_l^i p_{\pi_k(l)}^i = z_{j,k}^i  \qquad \forall j \in [\Theta], \forall k \in [K], \forall i \in [m].
$$

It is in the objective function where we distinguish which $z_{j,k}^i$ are relevant for which machine. In the following, let $z_j^i = z_{j,k}^i$ if machine $M_i$ is of kind $k$:
$$
f(\vex, \vez) =  \sum_{i=1}^m \sum_{j=1}^\Theta (\frac{1}{2} (z_j^i)^2 (\rho_i(j)-\rho_i(j+1)) + \frac{1}{2} x_j^i p_j^i w_j).
$$

Lower and upper bounds (\ref{eq:rcmax_ub})-(\ref{eq:rcmax_lb}) stay as before. Applying Theorem~\ref{thm:nfoldip_convex} concludes the proof of part \eqref{thm:rwici_it} of Theorem~\ref{thm:nfold_sched}.

\subsection{$R|| \sum w_j C_j$ parameterized by $m$ and $\theta$}

Finally, we examine the same scenario as before, but this time we restrict the number of machines $m$ to be a parameter, but, in turn, relax the restriction from $p_{max}, w_{max}$ to $\theta$. We use the same ILP formulation which Mnich and Wiese used to show that $R || C_{max}$ is \FPT with respect to $m$ and $\theta$. However, the careful analysis in Corollary~\ref{cor:separable_convex} was needed to show that the objective function is convex in order to apply Theorem~\ref{thm:semialgebraic_ip_fpt}.

Let $\Theta \leq \theta^m$ be the number of distinct types of jobs. We have variables $x_j^i$, $j \in [\Theta], i \in [m]$ and permutations $\pi_i$ with the same meaning as above. Notice that the following can be seen as a subset of the previous $n$-fold IP:
%\mkcom{this formatting is really ugly}
\begin{alignat*}{2}
& \text{minimize} \\
& f(\vex, \vez) =  \sum_{i=1}^m \sum_{j=1}^\Theta \biggl(\frac{1}{2} (z_j^i)^2 (\rho_i(j) &-&\rho_i(j+1)) + \frac{1}{2} x_j^i p_j^i w_j \biggr) \\
& \text{subject to} \\
& \sum_{i=1}^m x_j^i = n_j   &\forall& j \in [\Theta] \\
& \sum_{l=1}^j x_l^i p_{\pi_i(l)}^i = z_j^i   &\forall& j \in [\Theta], \forall i \in [m]
\end{alignat*}

%\begin{equation}
%\begin{aligned}
%& \text{minimize} & & f(\vex, \vez) =  \sum_{i=1}^m \sum_{j=1}^\Theta & (\frac{1}{2} (z_j^i)^2 (\rho_i(j) -\rho_i(j+1)) + \frac{1}{2} x_j^i p_j^i w_j) \\
%& \text{subject to}
%& & \sum_{i=1}^m x_j^i = n_j & \forall j \in [\Theta] \\
%& & & \sum_{l=1}^j x_l^i p_{\pi_i(l)}^i = z_j^i &  \forall j \in [\Theta], \forall i \in [m]
%\end{aligned}
%\end{equation}
In order to apply Theorem~\ref{thm:semialgebraic_ip_fpt}, we observe that both the number of variables $2 \Theta m$ and the number of constraints $\Theta + \Theta m$ are fixed parameters, and that the objective function $f$ is convex as shown in the previous subsection. This concludes the proof of Theorem~\ref{thm:rwicim}.

\section{Conclusions}
Although much is known about approximating scheduling problem, little is known from the parameterized complexity point of view about the most basic problems. The purpose of this paper is twofold. The first is to show new \FPT algorithms for some scheduling problems. The second is to demonstrate the use of $n$-fold integer programming, a recent and powerful variable dimension technique. We hope to encourage research in both directions. To facilitate this research, we point out the following open problems:

%%% better in itemize ??
%\textbf{(1)} Minimizing weighted flow time $P | r_j | \sum w_j F_j$ parameterized by $p_{max} + w_{max}$. \textbf{(2)} $P || C_{max}$ parameterized by $\theta$ instead of $p_{max}$. \textbf{(3)} Scheduling with precedence constraints $P | prec | C_{max}$ parameterized by $p_{max}$ and the partial order width of $prec.$ \textbf{(4)} $R | pmtn | \sum C_j$ parameterized by $m$ and $p_{max}$; this is justified by the problem being strongly \NP-hard~\cite{Sitters:05}, so parameterizing by $p_{max}$ is not enough. \textbf{(5)} $P || C_{max}$ parameterized by $p_{max}$ with both $m$ and $n$ given in binary; this might be possible using the recently developed result for huge $n$-fold IPs due to Onn and Sarrabezolles~\cite{OnnS:15}. Turning our attention to developing the techniques we use, we ask if \textbf{(6)} $n$-fold IP is FPT also for any separable convex functions (not just $p$-piecewise affine functions~\cite{HemmeckeOR:13}). We are also interested in \textbf{(7)} further applications of $n$-fold IP and \textbf{(8)} quasiconvex minimization over convex sets in fixed dimension.

\begin{itemize}
  \item Minimizing weighted flow time $P | r_j | \sum w_j F_j$ parameterized by $p_{max} + w_{max}$.
  \item $P || C_{max}$ parameterized by $\theta$ instead of $p_{max}$.
  \item $R | pmtn | \sum C_j$ parameterized by $m$ and $p_{max}$; this is justified by the problem being strongly \NP-hard~\cite{Sitters:05}, so parameterizing by $p_{max}$ is not enough.
  \item $P || C_{max}$ parameterized by $p_{max}$ with both $m$ and $n$ given in binary; this might be possible using the recently developed result for huge $n$-fold IPs due to Onn and Sarrabezolles~\cite{OnnS:15}.
  \item Multi-agent scheduling was studied by Hermelin et al.~\cite{HermelinKSTW15}; what results can be obtained by applying $n$-fold IP?
  \item Turning our attention to developing the techniques we use, we ask if $4$-block $n$-fold IP (a generalization of $n$-fold IP) is \FPT or \W{1}-hard; only an \XP algorithm is known so far~\cite{HemmeckeKW14}.
  \item We are also interested in further applications of $n$-fold IP and quasiconvex minimization over convex sets in fixed dimension.
\end{itemize}

\paragraph{Acknowledgements.}
We would like to thank René van Bevern for pointing us to much related work.

%\begin{acknowledgements}
%\end{acknowledgements}


\bibliographystyle{plain}
%\bibliography{sched}
\bibliography{arxiv}

\end{document}